\documentclass[copyright,creativecommons]{eptcs}
\providecommand{\event}{TARK 2025} 

\usepackage{iftex}
\usepackage{xcolor}
\usepackage{amsmath}
\usepackage{amssymb}
\usepackage{amsthm}
\usepackage[all]{xy}
\usepackage{xspace}
\usepackage{algpseudocodex}
\usepackage{faktor}

\newcommand{\putawayall}[1]{}

\newcommand{\putaway}[1]{}
\newcommand{\ijcaiputaway}[1]{}

\newcommand{\state}{S}
\newcommand{\stateval}{V}

\newcommand{\classbelbase}{\mathbf{M} }

\newcommand{\setbelbase}{\mathbf{S} }

\newcommand{\impbel}[2] {\square_{#1}^{#2}  }
\newcommand{\impbelposs}[2] {\lozenge_{#1}^{#2}  }

\newcommand{\belbaseset}{\mathcal{B}}

\newcommand{\iconstraint}{\mathit{U}}

\newcommand{\relstate}[1]{\mathcal{R}_{#1}}

\renewcommand{\phi}{\varphi}

\newcommand{\defin}{~\stackrel{\mbox{\scriptsize def}}{=}~} 

\newcommand{\imp}{\rightarrow} 
\newcommand{\eqv}{\leftrightarrow} 

\renewcommand{\phi}{\varphi}

\newcommand{\PROP}{\mathit{Atm}}

\newcommand{\AGT}{\mathit{Agt}}

\newcommand{\bnf}{::=}
\newcommand{\powerset}{ 2^}

\newcommand{\PAGT}{\powerset{\mathit{Agt}*}}

\newcommand{\lang}{ \mathcal{L}}

\newcommand{\fraglang}{ \mathcal{L}_{0} }

\newbox\itembox
\def\itemlistlabel#1{#1\hfill}
\def\itemlist#1{\setbox\itembox=\hbox{#1}%
                \list{}{\labelwidth\wd\itembox
                             \leftmargin\labelwidth
                             \advance\leftmargin by\itemindent
                             \advance\leftmargin by\labelsep
                             \let\makelabel\itemlistlabel}}

\newcommand{\algofunction}{\textbf{function }}

\newcommand{\algoprocedure}{\textbf{procedure }}

\newcommand{\algoendfunction}{\textbf{endFunction }}

\newcommand{\algofor}{\textbf{for }}
\newcommand{\algodo}{\textbf{do }}

\definecolor{algocommentbackgroundcolor}{rgb}{1,1,0.5}

\newcommand{\algowhile}{\textbf{while }}

\newcommand{\algoif}{\textbf{if }}
\newcommand{\algothen}{\textbf{then }}

\newcommand{\algoendif}{\textbf{endIf }}

\newcommand{\algomatch}{\textbf{match }}

\newcommand{\algocase}{\textbf{case }}

\newlength{\algoindentlongueur}
\newcommand{\algoindent}{\hspace*{\algoindentlongueur}}

\newlength{\algoindentavantvrulelongueur}
\newcommand{\algoindentavantvrule}{\hspace*{\algoindentavantvrulelongueur}}

\newlength{\dummy}

\newsavebox{\frameminipageboiteavecunnomsuperlongdesortequonnepuissepaslereutiliser}
\newenvironment{frameminipage}[2][c]{%
\begin{lrbox}{\frameminipageboiteavecunnomsuperlongdesortequonnepuissepaslereutiliser}%
\begin{minipage}[#1]{#2}%
} {%
\end{minipage}%
\end{lrbox}%
\framebox{\usebox{\frameminipageboiteavecunnomsuperlongdesortequonnepuissepaslereutiliser}}%
}

\newenvironment{algobloc}{\setlength{\dummy}{\linewidth}\addtolength{\dummy}{- \algoindentlongueur}\addtolength{\dummy}{- \algoindentavantvrulelongueur}\algoindentavantvrule\vrule\algoindent\begin{minipage}{\dummy}}{\end{minipage}}

\newtheorem{example}{Example}
\newtheorem{definition}{Definition}
\newtheorem{lemma}{Lemma}
\newtheorem{theorem}{Theorem}
\newtheorem{remark}{Remark}

\ifpdf
  \usepackage{underscore}         
  \usepackage[T1]{fontenc}        
\else
  \usepackage{breakurl}           
\fi

\newcommand{\refeq}[1]{(\ref{#1})}

\newcommand{\definitionword}{Def.}
\newcommand{\lemmaword}{Lem.}

\newcommand{\eqdef}{\stackrel{def}{=}}

\newcommand{\Props}{\mathbb{P}}
\newcommand{\prop}{p}
\newcommand{\ltri}[2]{\triangle_{#1}^{#2}}
\newcommand{\lbox}[2]{\Box_{#1}^{#2}}
\newcommand{\Lang}{\mathcal{L}}
\newcommand{\LangE}{\mathcal{L}_0}
\newcommand{\LogicName}{\textbf{LGDDA}\xspace}
\newcommand{\LDA}{\textbf{LDA}\xspace}
\newcommand{\CalculusName}{\textbf{Tab$_\LogicName$}\xspace}
\newcommand{\LKName}{\textbf{LK}\xspace}

\newcommand{\tuple}[1]{( #1 )}
\newcommand{\Nat}{\mathbb{N}}
\newcommand{\plusinfty}{{\omega}}
\newcommand{\NatZero}{\Nat_0}
\newcommand{\NatOne}{\Nat_1}
\newcommand{\NatZeroInf}{\Nat_0^\plusinfty}
\newcommand{\NatOneInf}{\Nat_1^\plusinfty}
\newcommand{\supp}[1]{Supp(#1)}
\newcommand{\maxhat}{\max^{\ast}}
\newcommand{\minhat}{\min^{\ast}}
\newcommand{\partitions}[2]{\mathcal{P}art(#2,#1)}

\newcommand{\Worlds}{W}

\newcommand{\dist}{\rho}
\newcommand{\Dox}{\mathcal{D}}
\newcommand{\Val}{\mathcal{V}}
\newcommand{\Agents}{\textit{Agt}}
\newcommand{\agent}{i}

\newcommand{\PowSet}[1]{2^{#1}}
\newcommand{\PowSetNE}[1]{2^{#1 \ast}}
\newcommand{\groups}{\PowSetNE{\Agents}}
\newcommand{\group}{J}
\newcommand{\MultiSet}[1]{\mathcal{M}(#1)}
\newcommand{\MultiSetE}{\mathcal{M}(\LangE)}
\newcommand{\summandf}{\delta}
\newcommand{\summand}[1]{\summandf(#1)}

\newcommand{\atomMap}{\chi}
\newcommand{\equivFilter}{\equiv_{\varphi}}

\newcommand{\completeSet}{\Phi}

\newcommand{\boxDown}[1]{#1^{\Box\downarrow}}

\newcommand{\ltriDown}[2]{#2^{\Delta_{#1}\downarrow}}
\newcommand{\sumbelief}[2]{\Dox_{agg}(#1, #2)}

\algnewcommand\BoolAnd{\textbf{and\;}}
\algnewcommand\BoolNot{\textbf{not\;}}
\algnewcommand\Continue{\textbf{continue}}

\newcommand{\rulename}[1]{\mbox{\textit{(#1)}}\xspace}

\newcommand{\ltrirule}{\rulename{$\ltri{}{}$-Mon}}
\newcommand{\lboxrule}{\rulename{$\lbox{}{}$-Elim}}

\title{Graded Distributed Belief}
\author{Emiliano Lorini
\institute{IRIT, CNRS, Toulouse University\\ Toulouse, France}
\email{Emiliano.Lorini@irit.fr}
\and
Dmitry Rozplokhas
\institute{TU Wien\\
Vienna, Austria}
\email{dmitry@logic.at}
}
\def\titlerunning{Graded Distributed Belief}
\def\authorrunning{E. Lorini \& D. Rozplokhas}
\begin{document}
\maketitle

\begin{abstract}
We introduce a new logic of graded distributed belief that allows us to express the fact that a group of agents distributively believe that a certain fact $\varphi$ holds with at least strength $k$. We interpret our logic by means of computationally grounded semantics relying on the concept of belief base. The strength of the group's distributed belief is directly computed from the group's belief base after having merged its members' individual belief bases. We illustrate our logic with an intuitive example, formalizing the notion of epistemic disagreement. We also provide a sound and complete Hilbert-style axiomatization, decidability result obtained via filtration, and a tableaux-based decision procedure that allows us to state PSPACE-completeness for our logic.
\end{abstract}

\section{Introduction}

The idea of using belief bases as formal semantics for multi-agent epistemic logic was first introduced in 
\cite{InPraiseOfBeliefBases}
 and further  developed in \cite{DBLP:journals/corr/abs-1907-09114,BeliefBasesAIJ}. This approach aligns with the sentential (or syntactic) perspective on knowledge representation \cite{KONOLIGE,HanssonJSL,ShohamJPL,JonghLiu2009}, which holds that an agent's body of knowledge
should be represented as a set of sentences in a formal language.
The key novelty of belief  base semantics, compared to traditional epistemic logic semantics based on multi-relational Kripke models \cite{MeyerEpist,DBLP:journals/ai/HalpernM92}, lies in two main aspects. First,  a possible world (or state) 
in a model 
is not treated as a primitive entity but is instead composed of the agents' belief  bases and a valuation of propositional atoms. Second, the agents' accessibility relations are not explicitly part of the model but are determined \emph{a posteriori} from their belief  bases. Specifically, in this semantics, an agent at state 
$S$ considers state 
$S'$
  possible (or epistemically accessible) if and only if 
$S'$
  satisfies all the formulas in the agent’s belief base at 
$S$.
This decomposition of a state into more fundamental elements is shared by various approaches in symbolic model checking and computationally grounded semantics for epistemic logic. These include frameworks based on interpreted systems, where a global state is decomposed into individual agents' local states \cite{Fagin1995,LomuscioRaimondi2015}, as well as those that rely on the primitive notion of observability or visibility \cite{DBLP:conf/aamas/HoekIW12,DBLP:conf/kr/CharrierHLMS16,DBLP:conf/lori/HerzigLM15,BenthemEijckGattingerSu2015}.
At the language level, the belief  base approach distinguishes explicit (or actual) belief from implicit (or potential) belief. The distinction between explicit and implicit belief has been widely discussed in the literature \cite{LevesqueExplicitBel,Fag87}. In the belief  base approach, this distinction is based on the concept of \emph{deducibility}: 
explicit beliefs are those directly stored in an agent's belief base, while implicit beliefs consist of any information that can be logically inferred from those explicit beliefs.

In two subsequent works, the belief base approach has been shown to successfully represent notions of distributed and common belief \cite{DBLP:conf/ecai/HerzigLPRS20,DBLP:conf/atal/LoriniR22}, as well as graded belief \cite{DBLP:conf/jelia/LoriniS21}.
On the one hand, belief base semantics allows for a natural distinction between explicit and implicit distributed belief. While the explicit distributed belief of a group is given by the merging of the belief bases of its members, the implicit distributed belief corresponds to what can be deduced from the (collective) belief base resulting from this merging.
On the other hand, the approach allows us to define a natural notion of the degree (or strength) of an agent's implicit belief that $\varphi$, understood as the maximum number of pieces of information that can be removed from the agent’s belief base without preventing the agent from deducing $\varphi$ from their explicit beliefs.

In this paper, we present a generalization of the belief base semantics for epistemic logic we  introduced in \cite{InPraiseOfBeliefBases,BeliefBasesAIJ}. In the original semantics, agents' belief bases were simply sets of formulas built from a language including propositional facts and explicit beliefs: an agent’s belief base could include both information about the world and information about other agents' belief bases.
Our generalization moves from plain (ungraded) belief bases to graded belief bases by using a multiset representation. In a graded belief base, each piece of information is associated with a natural number representing the strength of the agent’s explicit belief, with $ 0$
associated to a formula $\alpha $
meaning that the agent has no explicit belief that $\alpha$. 
We use this more general
semantics
to define a novel  notion
of \emph{graded distributed belief}, 
as a piece of information $\varphi $
that a group can deduce from
their collective belief
base with a given  strength $k$.
Given a graded belief base for each agent in a group, we compute the group’s graded distributed belief in two steps. First, we merge the graded belief bases of the group’s members to obtain a collective graded belief base. The degree that the group assigns to a formula 
$\alpha$
corresponds to the sum of the degrees that each member assigns to 
$\alpha$.
Second, we compute the group's degree of distributed belief in a certain fact $\varphi$ as the amount of information that can be removed from the group's belief base without preventing it from deducing $\varphi$. 
We also show how our framework
can be used to define a quantitative  notion of 
of epistemic disagreement  within a group, based on the amount of information that must be removed from the group’s collective belief base to restore consistency.
This notion of disagreement bears similarities to the measure of inconsistency defined in \cite{HunterKonieczny2008}, namely, the minimal number of formulas that need to be removed from a belief  base to restore consistency.

The paper is organized as follows. In Section \ref{sec:framework}, we first present the general framework: the graded belief base semantics, the modal language for representing implicit graded distributed belief,  its semantic interpretation
and an example illustrating our framework
as well as the notion of epistemic disagreement
within a group that can be defined in our framework. 
Following \cite{BeliefBasesAIJ}, in Section \ref{sec:alternative}, we introduce an alternative Kripke-style semantics for our modal language, which serves as a technical tool for investigating the proof-theoretic aspects of our framework.
Section
\ref{sec:Hilbert}
presents a Hilbert-style axiomatics
for our logic of   graded  distributed belief.
Then, in Section \ref{sec:tableaux} we present a decision procedure
based on tableaux, which allows us to establish PSPACE-completeness for our logic. 

Before turning to the core of the paper, we briefly discuss some related work. 
Although the notion
of graded distributed belief
and the graded belief base semantics used to interpret it, as introduced in this paper, are new, the notion of plain distributed belief has been widely investigated in epistemic logic \cite{DBLP:journals/ai/HalpernM92,DBLP:journals/synthese/WangA13,DBLP:journals/ai/AgotnesW17,RoelofsenJANCL,HLM1999,Lindqvist2022,Balbiani2024,Galimullin2024}. 
Moreover, 
the idea of having 
graded belief
modalities for individual
agents
was explored in previous work \cite{LavernyLang2004},
in 
line with work
in ranking theory \cite{Spohn1988}.
Other approaches employ graded modalities for individual agents, where the degree of belief is determined either by the number of worlds in which the believed formula holds true \cite{HoekMeyerGRADED,HoekThesis,Budzynska2008ALF}, or by the amount of evidence supporting it \cite{DBLP:conf/ijcai/BalbianiFHL19}.
The notion of graded belief
base
is also
used in possibility theory
\cite{Dubois1994}.
In the present paper,
we generalize
it to the multi-agent
setting and to nested beliefs.


\section{Framework}\label{sec:framework}

In this section, we
present our 
graded belief base semantics
and show how to use it  to compute
graded doxastic accessibility
relations for agents and groups. 
Then, we introduce a modal language
of graded distributed belief
and interpret it using the semantics.
We illustrate
our language
and semantics
with the help of a concrete example.

\subsection{Notation}

In this paper we will work with graded sets (or multisets), where grade (or cardinality) of each element is either a natural number (including zero) or infinity (denoted $\plusinfty$).
To avoid confusion we use notation $\NatZero$ and $\NatOne$ for natural numbers with and without zero respectively, and $\NatZeroInf$ and $\NatOneInf$ for their extensions with element $\plusinfty$.  We represent a graded set over set $X$ by a function ${f \colon X \to \NatZeroInf}$, and define the \emph{support} of this graded set as $\supp{f} = {\{ x \in X \mid f(x) > 0 \}} $. We denote the set of all multisets over $X$ by $\MultiSet{X}$.

For $X \subseteq \NatZeroInf$ we will use notation $\minhat X = \min (X \cup \{0\})$ and $\maxhat X = \max (X \cup \{\plusinfty\})$ to avoid dealing with the case of empty set.
We will also consider potentially infinite sums of grades. Since grades are natural numbers (or inifinity) such sums have natural well-behaved definition: we define it as sum of non-zero summands if there are finitely many such summands and none of them is $\plusinfty$, and as $\plusinfty$ otherwise.
Finally, we will use the notion of \emph{partitions}, i.e. functions dividing the given grade $k \in \NatZero$ among agents in group $\group$: $\partitions{k}{\group} = \{ \summandf : \group \to \NatZero \mid \sum_{\agent \in \group} \summandf(i) = k \}$.  

\subsection{Semantics}

We are going to present a 
belief base 
semantics for epistemic
attitudes of agents
that generalizes
the belief
base semantics
introduced in \cite{InPraiseOfBeliefBases,BeliefBasesAIJ}
to multisets.
Multisets are used
to represent agents' explicit  beliefs
with their strengths, weight
or epistemic importance. 
Unlike the standard Kripke semantics
for epistemic
logic
in which the notions
of epistemic alternative 
and plausibility of a world
(or state)
are given as primitive,
in this semantics
they are defined  from the primitive
concept
of graded belief base.

Assume
a countably  infinite set  of atomic propositions $\PROP$
and
 a finite set of agents $\AGT = \{ 1, \ldots, n \}$.
 The set of non-empty groups 
 is denoted by $\PAGT=
 2^\AGT \setminus \{\emptyset \}$.
We define the  language
$\fraglang(\PROP, \AGT)$
for
representing
agents' graded explicit beliefs
by the following  grammar:
\begin{center}\begin{tabular}{lcl}
  $\alpha$  & $\bnf$ & $p  \mid \neg\alpha \mid \alpha \wedge \alpha     \mid
  \ltri{\agent}{k} \alpha,
                        $
\end{tabular}\end{center}
where $p$ ranges over $\PROP$,
$\agent$ ranges over $\AGT$
and $k$
ranges over $\NatOneInf$.
The formula $ \ltri{\agent}{k} \alpha$ is read ``agent $\agent$ explicitly believes that $\alpha$
with at least degree  $k$''.
For notational convenience, we
abbreviate
$ \ltri{\agent}{} \alpha \eqdef
\ltri{\agent}{ 1 } \alpha$. 
The formula 
$ \ltri{\agent}{} \alpha$ is simply read ``agent $\agent$ explicitly believes that $\alpha$''.

For notational convenience we write
$\fraglang$  instead
 of $\fraglang(\PROP, \AGT)$, when the context is unambiguous.

\begin{definition}\label{state}
A state is a tuple $ S = (\belbaseset_1, \ldots, \belbaseset_n,  \stateval )$
where
 for every $\agent \in \AGT$, $\belbaseset_\agent \in \MultiSet{\fraglang}  $
is agent $\agent$'s graded  belief base,
and $ \stateval \subseteq \PROP  $ is the actual environment.
%
The set
of all states is denoted by $\setbelbase$.
\end{definition}


$\belbaseset_\agent$ assings a weight $\belbaseset_i(\alpha)$ to  each formula from $\alpha \in \fraglang$,
capturing the strength of agent $i$'s
explicit 
belief that $\alpha$.

The language
$\fraglang$
is interpreted with respect
to states, as follows.
\begin{definition}\label{truthcond1}
Let $S=(\belbaseset_1, \ldots,  \belbaseset_n,  \stateval ) \in \setbelbase$. Then:
\begin{eqnarray*}
S\models p & \Longleftrightarrow & p \in \stateval ,\\
S\models \neg \alpha & \Longleftrightarrow &    S \not \models  \alpha ,\\
S \models \alpha_1 \wedge \alpha_2 & \Longleftrightarrow &    S \models \alpha_1  \text{ and }     S \models \alpha_2, \\
S \models  \ltri{\agent}{k} \alpha   & \Longleftrightarrow & 
\belbaseset_i(\alpha)\geq k. 
\end{eqnarray*}

\end{definition}
Observe in particular the following interpretation
of the graded explicit belief operator: agent~$\agent$
explicitly believes that $\alpha$
with at least degree $k$
if and only if 
the information $\alpha$
has an importance for the agent 
 at least equal to $k$.

From the agents' graded belief
bases $\belbaseset_1, \ldots, \belbaseset_n$  it
is natural to 
compute
the
collective graded belief base
$\belbaseset_\group \in \MultiSet{\fraglang}$
of a group $\group$:
the 
degree of explicit belief of the
group  is equal to the sum of the degrees of beliefs of the group's members.
\begin{definition}\label{Merging}
Let $ S = (\belbaseset_1, \ldots, \belbaseset_n,  \stateval ) \in \setbelbase$
and $\group \in \PAGT$. 
Then,
$\belbaseset_\group (\alpha)= \sum_{\agent \in \group} \belbaseset_\agent ( \alpha)$ for every $\alpha \in \fraglang$.

\end{definition}

This definition extends the notion of ``pooling'' belief bases 
from~\cite{DBLP:conf/ecai/HerzigLPRS20} to the weighted case, assuming agent independence and epistemic egalitarianism. Note 
that we do not rely on more involved forms of belief base aggregation 
aimed at preserving consistency, since in our framework individual 
belief bases may already be inconsistent.

The following definition introduces  the notion of graded doxastic alternative
for a group.

\begin{definition}\label{DefAlternative}
Let $\group  \in \PAGT$
and let $k \in \NatZero$.
Then,
$\relstate{\group }^{k}$
is the binary relation on the set  $\setbelbase$
such that,
for all
$ S=(\belbaseset_1, \ldots,  \belbaseset_n, \stateval ) ,
  S' = (\belbaseset_1', \ldots,  \belbaseset_n',  \stateval' )  \in \setbelbase$:
\begin{align*}
&S \relstate{\group}^{k} S' \text{ if and only if } \sum_{\substack{\alpha \in \LangE \\ S' \not\models \alpha}} 
\belbaseset_\group  ( \alpha) \leq k.
\end{align*}
\end{definition}

$S \relstate{\group }^{k} S'$ means
that, from the point of view of the group  $\group $ at state $S$, 
 state $S'$ is at most $k$-implausible,
where the degree of implausibility 
of a state for a group 
is equal to the  weighted sum
of the group's explicit
beliefs
that are not satisfied at the state.
This means
that the degree
of implausibility
of a state for a group 
depends
on i) how much information 
that the group 
has in its belief base is not satisfied at the state,
and ii) how important 
is that information
for the group.

Notice that 
$S \relstate{\group }^{k} S'$ can also
be interpreted
as the fact that
state $\state'$
is considered possible
for the group $\group$
after removing from its collective belief
base 
a body of information
of importance at most equal to $k$. 
Indeed, $\sum_{\substack{\alpha \in \LangE \\ S' \not\models \alpha}} 
\belbaseset_\group  ( \alpha)$
can also
be conceived as the total amount 
of importance
for the group $\group$
at state $S$
of the information
that is not satisfied at state $S'$. 

A
graded doxastic accessibility relation $\relstate{\group}^{k} $
induces a plausibility ordering over states,
as in \cite{Spohn1988,LavernyLang2004}.
For notational convenience,
we write $\relstate{\group }^{}$
instead of
$\relstate{\group }^{0}$.
Clearly, 
$S \relstate{\group }^{} S'$
if and only iff
$\forall \alpha \in \fraglang $,
if $ \belbaseset_\group(\alpha)> 0$ then $S' \models \alpha$. 
In words,
a state is $0$-implausible
from the point of view of a group 
if it satisfies all information
in the group's belief base.

Before concluding this section,
we define the concept
of a model
as a
 state supplemented with a set of states,
called \emph{context}.
The latter includes all states 
 compatible with
the the
agents' common ground \cite{Stalnaker2002},
i.e., the body of information that the agents commonly believe to be the case.
\begin{definition}\label{MAGBM}
A multi-agent graded belief model (MAGBM) 
is a pair $ (S ,\iconstraint)$,
where
 $S \in \setbelbase $
and
$\iconstraint \subseteq \setbelbase$.
The class of models  is denoted by $\classbelbase$.
\end{definition}
\subsection{Language}

We consider a modal language 
 $\lang(\PROP, \AGT)$
 that extends  the language $\fraglang(\PROP, \AGT)$
given above
with graded distributed belief
modalities. It is 
 is defined by the following grammar:
\begin{center}\begin{tabular}{lcl}
  $\phi$  & $\bnf$ & $\alpha  \mid \neg\phi \mid \phi \wedge \phi   \mid \impbel{\group  }{k} \varphi ,
                        $\
\end{tabular}\end{center}
where $\alpha$ ranges over $\fraglang(\PROP, \AGT)$, $\group   $ ranges over $\PAGT$
and $k$ ranges over $\NatZero$.
For notational convenience we write
and  $\lang$  instead
of $\lang(\PROP, \AGT)$, when the context is unambiguous.
The other Boolean constructions  $\top$, $\bot$, $\imp$ and $\eqv$ are defined in the standard way.
%
%

We interpret the modal language 
$\lang$
relative to a model
by means
of the graded accessibility relations
 $\relstate{\group}^{k}$.
\begin{definition}\label{truthcond2}
Let $ (S,\iconstraint)  \in \classbelbase$. Then:
\begin{eqnarray*}
 (S,\iconstraint) \models \alpha & \Longleftrightarrow & S \models \alpha, \\
(S,\iconstraint) \models \neg \varphi & \Longleftrightarrow & S \not\models \varphi, \\
(S,\iconstraint) \models \varphi_1 \wedge \varphi_2 & \Longleftrightarrow & S \models \varphi_1 \textit{ and } S \models \varphi_2, \\
(S,\iconstraint) \models \impbel{\group }{k} \varphi & \Longleftrightarrow & \forall S' \in  \iconstraint  , \text{ if } S \relstate{\group }^{k}S' \text{ then }  ( S' , \iconstraint) \models \varphi.
\end{eqnarray*}
\end{definition}
The 
modal 
formula $ \impbel{\group  }{k}   \varphi$ is read
``group  $\group $ would implicitly believe that $\varphi$, 
for every removal
from its  belief base
of a body of information
of importance at most equal to $k$''. The value $k$
can also be conceived as the extent to which group 
$\group $
distributively 
believes that $\varphi$.
Indeed, the higher the 
importance of the information 
that can be removed
from the group's   belief base 
without affecting  what the group can infer,
the stronger the inference and so the group's resulting  distributed belief.
Thus, $ \impbel{\group }{k}   \varphi$ can also be read
``group  $\group $
has an implicit distributed
belief  that $\varphi$ of degree (or strength) at least $k$''.
%
The abbreviation $\impbelposs{\group }{k}   \varphi \defin \neg\impbel{\group }{k} \neg \varphi$ defines the concept
of distributed belief compatibility. 
The formula 
$\impbelposs{\group }{k}   \varphi$
has to be read
``$\varphi$
would be compatible with group  $\group $'s
explicit beliefs,
for some removal 
from its collective belief base
of a body of information
of importance at most equal to $k$''.

\subsection{Conceptual Analysis}

We are going to show how our language and semantics
can be leveraged to formally represent
graded
distributed belief 
as well as  degree
of epistemic disagreement within  a group of agents.
The latter notion is formally defined by the following abbreviation:
$\mathit{Disagree} (\group  ,k ) \eqdef \impbel{\group }{k-1 } \bot $ for $k\ge1$.
$   \mathit{Disagree} (\group  ,k ) $
means that within  the group $\group$
there is an epistemic disagreement
of at least strength $k$.

\begin{example}
Ann, Bob, Cath and John are the four members of a research project evaluation committee. Their task is to decide whether a submitted project for funding can be included in the list of ``fundable''  projects or not. They 
are all convinced with at least strength $k_0>0$
that a 
project should be included in the list ($\mathit{in}$) if and only if its idea is innovative 
($\mathit{id}$) 
and, at the same time, the project's consortium is of high scientific standard ($\mathit{hi}$). 
This hypothesis is captured by the following abbreviation:
\begin{align*}
    \alpha_1 \eqdef 
\bigwedge_{i \in \{ \mathit {Ann},\mathit {Bob},\mathit {Cath},\mathit {John} \} }  \ltri{i}  {k_0 } \big( 
\mathit{in} \leftrightarrow ( \mathit{id} \wedge \mathit{hi} )
\big) .
\end{align*}
However, they have diverging opinions and, in some cases, have not yet formed an opinion regarding these qualities of the project.
In particular, Ann explicitly believes with degree $k_1 > 0$ that the project's idea is innovative, and Cath believes the opposite with degree $k_3 > 0$, Bob explicitly believes with degree $k_2 > 0$ that the project's consortium is of high scientific standard, and John explicitly believes the opposite with degree $k_4 > 0$.
This hypothesis is captured by the following abbreviation:
\begin{align*}
    \alpha_2 \eqdef &
 \ltri{\mathit {Ann } }  { k_1}  \mathit{id} \wedge 
  \ltri{\mathit {Bob } }  { k_2}   \mathit{hi }\wedge
   \ltri{\mathit {Cath } }  { k_3}   \neg \mathit{id } \wedge 
     \ltri{\mathit {John } }  { k_4 } \neg   \mathit{hi }
   .
\end{align*}

It is routine to verify that group $\{ \mathit{Ann}, \mathit{Bob} \}$ implicitly believes with degree $(\min\{2 k_0, k_1, k_2\} - 1)$ that the project should be included in the list, while group $\{ \mathit{Cath}, \mathit{John} \}$ believes the opposite with degree $(\min\{2 k_0, k_3 + k_4\} - 1)$:
\[\models (\alpha_1 \wedge \alpha_2) \;\;\rightarrow\;\;  (\impbel{\{ 
     \mathit{Ann  } , 
    \mathit{Bob  }   \} }{ \min\{2 k_0, k_1, k_2\} - 1 } \mathit{in } \;\;\;\wedge\;\;\; \impbel{\{ 
     \mathit{Cath  } , 
    \mathit{John   }   \} }{\min\{2 k_0, k_3 + k_4\} - 1  } \neg \mathit{in }).\]

Moreover, when the explicit information is restricted to $\alpha_1 \wedge \alpha_2$ there is no disagreement within  these groups (i.e. there exist models where $\alpha_1 \wedge \alpha_2$ is true and $\mathit{Disagree}(\group,1)$ is false), while all four agents together have an epistemic disagreement of at least strength $(\min\{k_1,k_3\} + \min\{k_2,k_4\})$:
\[ \begin{array}{l}
\not\models (\alpha_1 \wedge \alpha_2)  \;\;\rightarrow\;\; (\mathit{Disagree} (\{\mathit{Ann},\mathit{Bob} \}  ,  1 ) \;\vee\; \mathit{Disagree} (\{\mathit{Cath},\mathit{John} \}  ,  1 )) \\
\models (\alpha_1 \wedge \alpha_2)  \;\;\rightarrow\;\; \mathit{Disagree} (\{\mathit{Ann},\mathit{Bob}, \mathit{Cath}, \mathit{John} \}  ,  \min\{k_1,k_3\} + \min\{k_2,k_4\}  ).
\end{array} \]


\end{example}

\section{Alternative  semantics}\label{sec:alternative}

To explore the proposed logic we will follow the general approach of~\cite{InPraiseOfBeliefBases} and introduce two alternative equivalent semantical characterizations: \emph{notional} and \emph{quasi-notional graded doxastic models}.

First, we redefine models in terms closer to Kripke semantics for modal logics by introducing the notion of distance between states for a given group of agents as the sum of degrees in the first world of all beliefs not satisfied in the second world for all agents in the group. Note that this notion of distance is not (in general) symmetric.

\begin{definition}
\label{def:NGDM}
A \emph{notional graded doxastic model (NGDM)} is a tuple $M = \tuple{\Worlds, \Dox, \dist, \Val}$ where $\Worlds$ is a set of worlds, ${\Dox \colon \Agents \times \Worlds \to \MultiSetE}$ is a doxastic function, ${\dist \colon \groups \times \Worlds \times \Worlds \to \NatZeroInf}$ is a distance function, and ${\Val \colon \Props \to \PowSet{\Worlds}}$ is a valuation, such that:
\begin{equation}
\tag{NGDM-DOX}
\label{eq:NGDM_dist_dox}
\dist(\group, w, u) = \sum_{\substack{\alpha \in \LangE \\ (M,u) \not\models \alpha}} \sum_{\agent \in \group} \Dox(\agent, w)( \alpha)
\end{equation}
with satisfaction relation defined as follows:

\begin{tabular}{lll}
$(M,w) \models \prop$ & $\Leftrightarrow$ & $w \in \Val(\prop)$  \\
$(M,w) \models \neg \varphi$ & $\Leftrightarrow$ & $(M,w) \not\models \varphi$  \\
$(M,w) \models \varphi_1 \land \varphi_2$ & $\Leftrightarrow$ & $(M,w) \models \varphi_1$ and $(M,w) \models \varphi_2$ \\
$(M,w) \models \ltri{\agent}{k} \alpha$ & $\Leftrightarrow$ & $\Dox(\agent, w)(\alpha) \geq  k$  \\
$(M,w) \models \lbox{\group}{k} \varphi$ & $\Leftrightarrow$ & $\forall u \in \Worlds: \dist(\group, w, u) \le k \Rightarrow (M,u) \models \varphi$  \\
\end{tabular}

$M$ is called \emph{finite} when $W$ is finite and $\supp{\Dox(\agent, w)}$ is finite for every $\agent$ and $w$.
\end{definition}

Notice that condition~\refeq{eq:NGDM_dist_dox} fully defines distance function $\rho$ via doxastic function $\Dox$ and the resulting distance function is \emph{additive on groups}:
\begin{equation}
\tag{NGDM-$\rho$-ADD}
\label{eq:NGDM_dist_add}
\dist(\group, w, u) = \sum_{\agent \in \group} \dist(\{\agent\}, w, u)
\end{equation}




However, condition~\refeq{eq:NGDM_dist_dox} itself is not axiomatizable (see Remark~\ref{rem:NDM_axiomatizability} below), therefore we also introduce the notion of \emph{quasi-notional graded doxastic models (QNGDMs)} where this condition is weakened to become axiomatizable. First, instead of equality, we require that the distance is no smaller than the sum of 
degrees of the unsatisfied beliefs (condition~\refeq{eq:qNGDM_dist_dox}). Additionally, we require a weakened version of \refeq{eq:NGDM_dist_add} (condition~\refeq{eq:qNGDM_dist_add}): the finite distance for every group $\group$ can be partitioned into summands $\summand{\agent}$ for every agent $\agent \in \group$ such that $\summand{\agent}$ is at least the distance for any agent $\agent \in \group$, and moreover the distance for any sub-group does not exceed the sum of $\summand{\agent}$ for the agents involved.

\begin{definition}
\label{def:QNGDM}
A \emph{quasi-notional graded doxastic model (QNGDM)} is a tuple $M = \tuple{\Worlds, \Dox, \dist, \Val}$ where $\Worlds$, $\dist$, $\Dox$, $\Val$ are as in \definitionword~\ref{def:NGDM} except that \refeq{eq:NGDM_dist_dox} is replaced by the two following weaker conditions for every $\group \in \groups$ and $w,u \in W$ if $\dist(\group, w, u) \neq \plusinfty$:
\begin{equation}
\tag{QNGDM-DOX}
\label{eq:qNGDM_dist_dox}
\dist(\group, w, u) \ge \sum_{\substack{\alpha \in \LangE \\ (M,u) \not\models \alpha}} \sum_{\agent \in \group} \Dox(\agent, w)( \alpha) 
\end{equation}
\begin{equation}
\tag{QNGDM-$\rho$-ADD}
\label{eq:qNGDM_dist_add}
\exists \summandf \in \partitions{\group}{\dist(\group, w,u)} \textit{ such that for any non-empty } \group' \subset \group,  \sum_{\agent \in \group'} \summand{\agent} \ge \dist(\group', w, u)
\end{equation}
\end{definition}

We will now show equivalence of the three semantics following the strategy of~\cite{BeliefBasesAIJ}: we will first show that semantics of QNGDMs satisfies the finite model property (\lemmaword~\ref{lem:QNGDM_FMP}), then we will show how every satisfying \emph{finite} QNGDM can be transformed into a satisfying NGDM (\lemmaword~\ref{lem:QNGDM_to_NGDM}), satisfying NGDM --- into a satisfying MAGBM (\lemmaword~\ref{lem:NGDM_to_MAGBM}), and a satisfying MAGBM into a satisfying QNGDM (\lemmaword~\ref{lem:MAGBM_to_QNGDM}), closing the circle. 

To establish Finite Model Property for QNGDMs we adapt the standard filtration technique.

\begin{lemma}
\label{lem:QNGDM_FMP}
If $M$ is QNGDM satisfying $\varphi \in \Lang$ then there exists a finite QNGDM $M'$ satisfying $\varphi$.
\end{lemma}
\begin{proof}{\it(Sketch)}
We consider an equivalence relation $\equivFilter$ on worlds of $M$, relating worlds with the same evaluation on all subformulas of $\varphi$. Then we transform $M = \tuple{\Worlds,\dist,\Dox,\Val}$ into finite a QNGDM ${\tuple{\Worlds',\dist',\Dox',\Val'}}$ as follows: $\Worlds' = \faktor{\Worlds}{\equivFilter}$; $\dist'(\group, U, V) \eqdef \min \{ \dist(\group, u, v) \mid {u \in U, v \in V}  \}$; $\Dox'(\agent, U)(\alpha) \eqdef {\maxhat \{ k \mid {\ltri{\agent}{k} \alpha \textit{ is a subformula of }\varphi \textit{ and } \Dox(\agent, u)(\alpha) \ge k \textit{ for all $u \in U$} \}}}$; $\Val'(\prop) \eqdef \{U \mid U \subseteq \Val(p) \}$. We check that $M'$ preserves evaluation on subformulas of $\varphi$ and satisfies both conditions of QNGDMs.
\end{proof}

Now we show how to transform a finite QNGDM into a (finite) NGDM. We will adapt the idea of two-stage model tranformation for distributed belief bases from~\cite{DBLP:conf/ecai/HerzigLPRS20} to our graded setting. At the fist stage we achieve \eqref{eq:NGDM_dist_add} in a QNDM by creating copies of each world for every possible group and redefining distances for them on the basis of partitions given by condition~\eqref{eq:qNGDM_dist_add}. At the second stage we achieve \eqref{eq:NGDM_dist_dox} by adapting the transformation from~\cite{DBLP:conf/jelia/LoriniS21} to our case: introducing a fresh characterizing atom for each world 
and add such atoms as beliefs with the required degree to satisfy the equality in \eqref{eq:NGDM_dist_dox}. 


\begin{lemma}
\label{lem:QNGDM_to_NGDM}
If $M$ is a finite QNGDM satisfying $\varphi \in \Lang$ then there exists a finite NGDM $M''$ satisfying $\varphi$.
\end{lemma}
\begin{proof}{\it(Sketch)}
We first change the set of worlds $W$ in $M$ to $W' = W \times \groups$, keep $\Val$ and $\Dox$ the same for each copy, and redefine distances for copies using $\summandf$ from condition~\eqref{eq:qNGDM_dist_add} for $\group$, $w$ and $u$: $\dist'(\group', (w,\group''), (u,\group)) = \sum_{\agent \in J'}\summand{\agent}$ for $\group' \subseteq \group$ and $\dist'(\group', (w,\group''), (u,\group)) = \plusinfty$ otherwise. With such definition of distances, the transformed model $M'$ preserves the satisfaction relation for each copy, trivially satisfies 
\eqref{eq:NGDM_dist_add}, and still satisfies \eqref{eq:qNGDM_dist_dox}, since distances between copies did not decrease w.r.t. original distances (by condition~\eqref{eq:qNGDM_dist_add}).
At the second stage, for each $w' \in W'$ we select a distinctive atom $\atomMap(w')$ not appearing in $\varphi$ and in any $\alpha \in \supp{\Dox'(i,w'})$ in $M'$ (which we can do since $M'$ is finite), and change $\Val'$ to $\Val''$ such that $\Val''(\chi(w')) = W' \setminus \{ w' \}$ (not changing the valuation on other atoms). Then we change the degrees of these atoms: $\Dox'(\agent, w)(\atomMap(u)) =
 \dist(\agent, w, u) - \sum\limits_{\substack{\alpha \in \LangE \\ (M',u) \not\models \alpha}} \Dox(\agent, w)(\alpha)$, turning inequality in \eqref{eq:qNGDM_dist_dox} into equality for $M'' = \tuple{\Worlds',\dist',\Dox'',\Val''}$.
\end{proof}

\begin{remark}
\label{rem:NDM_axiomatizability}
Contrapositive reading of Lemma~\ref{lem:QNGDM_to_NGDM} implies that any formula from $\lang$ satisfied in all NGDMs is also satisfied in all QNGDMs, which is a proper superclass of NDMs. Therefore there can be no characterizing axiom $\xi \in \lang$ that is true exactly in the QNGDMs satisfying \eqref{eq:NGDM_dist_dox}. On the other hand, as we will see in the next section, there is a characterizing axioms in this sense for properties \eqref{eq:qNGDM_dist_dox} and \eqref{eq:qNGDM_dist_add}.
\end{remark}

We can also easily transform a satisfying NGDM into a satisfying MAGBM.

\begin{lemma}
\label{lem:NGDM_to_MAGBM}
If $\varphi$ is satisfied by some NGDM then $\varphi$ is satisfied by some MAGBM.
\end{lemma}
\begin{proof}{(Sketch)}
Each world can be mapped into a state (by reconstructing belief bases from the doxastic function) and the evaluation will be preserved thanks to the condition \eqref{eq:NGDM_dist_dox}.
\end{proof}

And, finally, we can move from MAGBMs back to QNGDMs, closing the cycle.

\begin{lemma}
\label{lem:MAGBM_to_QNGDM}
If $\varphi$ is satisfied by some MAGBM then $\varphi$ is satisfied by some QNGDM.
\end{lemma}
\begin{proof}{(Sketch)}
The doxastic function is defined by belief bases, and the distances can be defined via condition \eqref{eq:NGDM_dist_dox}, apart from the distances to the initial state $S$ defined as $\plusinfty$ (to reflect that, in general, $S$ does not belong to the context $U$). 
\end{proof}

Thus, we have established the equivalence of all three introduced semantics (MAGBMs, NGDMs, and QNGDMs), which will be crucial for axiomatization of the proposed logic. Moreover, the described transformations turn an arbitrary QNGDM into an NGDM of exponential size, which implies decidability of the proposed logic (since all NGDMs of bounded size can be checked for satisfaction in finite time).

\section{Axiomatics}\label{sec:Hilbert}

In this section, we present an axiomatization for the proposed logic 
\LogicName
(Logic of Graded Distributed Doxastic Attitudes) 
based on the semantics of QNGDMs introduced in the previous section.

\begin{definition}
\label{def:axioms}
Logic \LogicName extends the classical propositional logic by the following axioms and rules:
\begin{align*}
\frac{\varphi}{\lbox{\group}{k} \varphi} & & \tag{$\mathbf{Nec}_{\lbox{\group}{k}}$}\label{eq:box_nec_rule} \\
\lbox{\group}{k} (\varphi \to \psi) \to (\lbox{\group}{k} \varphi \to \lbox{\group}{k} \psi) & & \tag{$\mathbf{K}_{\lbox{\group}{k}}$}\label{eq:box_K_axiom} \\ 
\ltri{\agent}{k} \alpha \to \ltri{\agent}{k'} \alpha & \qquad \textit{if } k \ge k' & \tag{$\mathbf{Mon}_{\ltri{\agent}{l}}$}\label{eq:tri_mon_axiom} \\
(\bigwedge_{\ltri{\agent}{k'} \alpha \in \Omega} \ltri{\agent}{k'} \alpha) \to \lbox{\group}{k} \!\!\!\!\!\!\!\!\bigvee_{\substack{\Omega' \subseteq \Omega \\  \textit{Sum}(\Omega') \le k}} \bigwedge_{\ltri{\agent}{k'} \alpha \in \Omega \setminus \Omega'} \!\!\!\!\alpha & \qquad\begin{array}{l}
\textit{if } \Omega \subseteq \{ \ltri{\agent}{k'} \alpha \mid \agent \in \group \} \\ \textit{and } \ltri{\agent}{k'} \alpha, \ltri{\agent}{k''} \alpha \in \Omega \Rightarrow k'  = k'' \\ \textit{where $\textit{Sum}(\Omega') = \sum_{\ltri{\agent}{k'} \!\!\alpha \in \Omega'} k'$} \end{array} & \tag{$\mathbf{Int}_{\ltri{\agent}{k'},\lbox{\group}{k}}$}\label{eq:tri_box_int_axiom} \\
(\bigwedge_{\lbox{\group'}{k'} \psi \in \Psi} \lbox{\group'}{k'} \psi) \to \lbox{\group}{k} \bigvee_{\summandf \in \partitions{\group}{k}} \bigwedge_{\substack{\lbox{\group'}{k'} \psi \in \Psi \\ \sum_{\agent \in \group'} \summandf(\agent) \le k'}} \!\!\!\!\!\!\!\!\!\!\!\!\!\psi 
&  \qquad \textit{if } \Psi \subseteq \{ \lbox{\group'}{k'} \psi \mid \group' \subseteq \group \}   & \tag{$\mathbf{Int}_{\lbox{\group'}{k'},\lbox{\group}{k}}$}\label{eq:box_box_int_axiom} \\
\end{align*}

\vspace{-7mm}\noindent $\varphi \in \Lang$ is said to be derivable from $\Gamma \subseteq \Lang$ in \LogicName (denoted $\Gamma \vdash_\LogicName \varphi$) when there is finite $\Gamma_f \subseteq \Gamma$ such that formula $(\bigwedge_{\psi \in \Gamma_f} \psi) \to \varphi$ can be derived using axioms and rules of \LogicName.
\end{definition}

Rule \eqref{eq:box_nec_rule} and axiom \eqref{eq:box_K_axiom} reflect the fact that $\lbox{\group}{k}$ is a normal modality. Monotonicity axiom \eqref{eq:tri_mon_axiom} reflects the fact that $k$ in modality $\ltri{\agent}{k}$ gives a lower bound on the weight of the belief. Two final axioms determine the interaction between triangles and boxes, and between boxes with different groups and degrees respectively. Informally, axiom \eqref{eq:tri_box_int_axiom} states that group $\group$ believes with level $k$ that their pulled beliefs are correct apart from some subset with the cummulative importance not exceeding $k$.
Axiom \eqref{eq:box_box_int_axiom} captures the fact that the distance $d$ between states (as defined by \eqref{eq:NGDM_dist_dox}) for group $\group$ can be partitioned into distances $\summand{\agent}$ for $\agent \in \group$ such that any graded belief for a subgroup $\group'$ with degree $k'$ is preserved for group $\group$ with degree $k \ge d$, as long as $k'$ is greater than the sum of distances $\summand{\agent}$ for $\agent \in \group'$. Notice that the following two validities, defining monotonicity of boxes w.r.t. group and degrees, are instances of \eqref{eq:box_box_int_axiom} for the case when there is only one box on the left:
\begin{align*}
\lbox{\group}{k} \varphi \to \lbox{\group}{k'} \varphi & \qquad \textit{if } k \ge k' & \tag{$\mathbf{Mon}^k_{\lbox{\group}{k}}$}\label{eq:box_mon_k_axiom} \\
\lbox{\group}{k} \varphi \to \lbox{\group'}{k} \varphi & \qquad \textit{if } \group \subseteq \group' & \tag{$\mathbf{Mon}^\group_{\lbox{\group}{k}}$}\label{eq:box_mon_group_axiom}
\end{align*}
\begin{remark}
Note that the axiomatization in \definitionword~\ref{def:axioms} does not simply adjust and merge axioms for graded beliefs and distributed beliefs. While \eqref{eq:tri_box_int_axiom} is a natural adjustment of axiom $(\mathbf{Int}_{\Delta_i,\Box_i})$ from~\cite{DBLP:conf/jelia/LoriniS21} for the case of graded belief bases, the logic of distributed belief from~\cite{DBLP:conf/ecai/HerzigLPRS20} requires only axiom \eqref{eq:box_mon_group_axiom}, and our logic for distributed graded belief requires significantly more sophisticated (and perhaps less intuitive) axiom \eqref{eq:box_box_int_axiom} that reflects the combinatorics of partition of distances for groups and subgroups captured in condition~\eqref{eq:qNGDM_dist_add}.  
\end{remark}

The proposed axiomatization is sound and complete w.r.t. QNGDM model semantics. 
To prove completeness we construct a canonical QNGDM, using rule \eqref{eq:box_nec_rule} and axioms \eqref{eq:box_K_axiom} and \eqref{eq:tri_mon_axiom} to establish the truth lemma, and axioms \eqref{eq:tri_box_int_axiom} and \eqref{eq:box_box_int_axiom} characterizing conditions \eqref{eq:qNGDM_dist_dox} and \eqref{eq:qNGDM_dist_add} of QNGDMs respectively.

\begin{theorem}
\label{th:axioms_sound_complete}
$\Gamma \vdash_{\LogicName} \varphi$ iff $\Gamma \vDash_{QNGDM} \varphi$ (i.e. any pointed QDNM satisfying all formulas from $\Gamma$ also satisfies $\varphi$).
\end{theorem}
\begin{proof}{\it(Sketch)}
\textbf{Soundness:} Validity of axiom \eqref{eq:tri_box_int_axiom} in any QNGDM follows from the condition~\eqref{eq:qNGDM_dist_dox} (if $\dist(\group, w, u) \le k$ then
the disjunct for $\Omega' = \{\ltri{\agent}{k'} \alpha \in \Omega \mid M, u \not\models \alpha\}$ will be satisfied in $u$) and validity of axiom \eqref{eq:tri_box_int_axiom} --- from condition~\eqref{eq:qNGDM_dist_add} (it ensures existence of partition $\summandf$ that corresponds to a satisfied disjunct), validity of other axioms and preservation of validity by the rules is trivial.

\textbf{Completeness:} We construct a canonical model 
$M^C = \tuple{\Worlds^C, \Dox^C, \dist^C, \Val^C}$, where $\Worlds^C$ is the set of all \LogicName-maxiconsistent sets of formulas (i.e. maximal sets not deriving $\bot$), $\Dox^C(\agent, \completeSet)(\alpha) = \maxhat \{ l \colon \ltri{\agent}{l} \alpha \in \completeSet \}$; $\dist^C(\group, \completeSet, \completeSet') = \minhat \{ k \colon \lbox{\group}{k} \varphi \in \completeSet \Rightarrow \varphi \in \completeSet' \}$; $\Val^C(\prop) = \{ \completeSet \in \Worlds^C \colon \prop \in \completeSet \}$. The truth lemma stating $M^C, \completeSet \models \psi \Leftrightarrow \psi \in \completeSet$ is proved by structural induction 
on $\psi$, axiom \eqref{eq:tri_mon_axiom} is used for the case of triangles, and the case of box follows from presence of \eqref{eq:box_K_axiom} and \eqref{eq:box_nec_rule} as usual. Then conditions \eqref{eq:qNGDM_dist_dox} and \eqref{eq:qNGDM_dist_add} can be proved contrapositively using axioms \eqref{eq:tri_box_int_axiom} and \eqref{eq:box_box_int_axiom} respectively. Failure of \eqref{eq:qNGDM_dist_dox} for some $w$, $u$, and $\group$ would allow us to select some finite set $\Omega$ of beliefs false in $u$ with aggregated weight in $w$ greater than $\dist(\group, w, u)$, and applying axiom \eqref{eq:tri_box_int_axiom} we would conclude that some beliefs in $\Omega$ are true in $u$. Failure of \eqref{eq:qNGDM_dist_add} for some $w$, $u$, and $\group$ would imply existence of subgroup $\group_\summandf \subset \group$ for every $\summandf \in \partitions{\group}{\dist^C(\group, w, u)}$ such that $d_\summandf = \sum_{\agent \in \group_\delta} \summand{\agent} < \dist^C(\group_\delta, \completeSet, \completeSet')$, which also implies existence of $\psi_\summandf \in \Lang$ such that $\lbox{\group_\summandf}{d_\summandf} \psi_\summandf$ is satisfied in $w$ but $\psi_\summandf$ is not satisfied in $u$. Taking $\Omega = \{ \lbox{\group_\summandf}{d_\summandf} \psi_\summandf \mid \summandf \in \partitions{\group}{\dist^C(\group, w, u)} \}$ and applying axiom \eqref{eq:tri_box_int_axiom} to it, we can derive that some $\psi_\summandf$ is satisfied in $u$, reaching contradiction.
\end{proof}

\section{Tableau calculus and satisfiability checking}
\label{sec:tableaux}

In this section, we present a tableau-based decision procedure for our logic and establish PSPACE-completeness of satisfiability checking (the same complexity as the basic logic of belief bases in \cite{BeliefBasesAIJ}).

\begin{definition}
The tableau calculus \CalculusName extends the standard tableau calculus \LKName for the classical logic with the following two rules:

\vspace{3mm}
\begin{tabular}{cr}
$\dfrac{\{ \ltri{\agent}{k} \alpha, \neg \ltri{\agent}{k+t} \alpha \} \cup X}{\{\bot\}}$ \ltrirule & $\dfrac{\{ \neg \lbox{\group}{k} \varphi \} \cup X}{ \{ \neg \varphi \} \cup \boxDown{Y_1} \cup \ltriDown{\group}{Y_1} \;|\; \dots \;|\; \{ \neg \varphi \} \cup \boxDown{Y_N} \cup \ltriDown{\group}{Y_N}}$ \lboxrule \\ 
\end{tabular}
\[\begin{array}{lrll}
\textit{where } & \{ Y_1, \dots, Y_N\} & = \{ Y \subseteq X \mid \exists \summandf \in \partitions{\group}{k}: & \;  \lbox{\group'}{k'} \psi \in Y \Rightarrow \sum_{\agent \in \group'}{\summandf(\agent) \le k'} \textit{ and }
\\ & & & \sumbelief{\agent}{X \setminus Y} \le \summandf(\agent) \; \forall \agent \in \group \}
\\ &  \sumbelief{\agent}{X \setminus Y} & \multicolumn{2}{l}{= \sum_{\alpha \in \LangE} \maxhat \{ k' \colon \ltri{\agent}{k'} \alpha \in X \setminus Y \} }
\\ & \boxDown{Y} \cup \ltriDown{\group}{Y} & \multicolumn{2}{l}{= \{ \psi \mid \lbox{\group'}{k'} \psi \in Y \} \cup \{ \alpha \mid \ltri{\agent}{k'} \alpha \in Y, \agent \in \group \} }
\end{array}\]
\end{definition}

Each tableau rule states that the satisfiability of a set of formulas above the line (called \emph{numerator}) implies satisfiability of at least one of the sets of formulas below the line separated by the symbol `|' (called \emph{denominators}). We can use it to derive the non-satisfiability of some formula by applying the rules sequentially, with each branch ending with a set containing $\bot$ (thus unsatisfiable). Such tree-like derivations are called \emph{closed tableaux}. For our logic, it is sufficient to add just one rule for each modality (in addition to standard rules for the classical connectives). The rule \ltrirule captures the monotonicity of triangles w.r.t. grades, and the rule \lboxrule adapts the rule eliminating negative boxes to our logic: if $\{\neg \lbox{\group}{k} \varphi\} \cup X$ is satisfied in some world $w$ of some model $M$ then there should exist a world $u$ with $\dist(\group, w, u) \le k$ and a partition $\summandf \in \partitions{\group}{k}$, such that $u$ satisfies $\neg \varphi$, all boxed formulas for subgroups with smaller distance, and some selection of triangled formulas in $X$, such that the cumulative degree of the rest of triangled formulas does not exceed $k$. Note that the rule \lboxrule has potentially exponentially many denominators. At the same time, in all rules (apart from the closing ones) the denominators contain only subformulas of formulas in the numerator, and the total number of connectives and modalities in each denominator decreases w.r.t. the numerator, so the length of each branch is linear w.r.t. size of the initial formula.
The resulting calculus is sound and complete w.r.t. logic \LogicName.

\begin{theorem}
\label{th:calculus_sound_complete}
A closed tableau in \CalculusName starting from  $\{ \varphi \}$ exists iff $\varphi$ is not satisfiable in \LogicName.
\end{theorem}
\begin{proof}{\it(Sketch)}
The soundness of all rules w.r.t. QNGDM semantics is straightforward. To prove completeness we show that for each set $\Gamma$ underivable in \CalculusName (i.e. there are no closed tableaux starting from $\Gamma$) there is a pointed QNGDM model $(M, w_0)$ satisfying all the formulas in $\Gamma$. We do it by induction on the number of connectives and modalities in $\Gamma$. First, if $\Gamma$ contains formulas with double negation, conjunction or negated conjunction on the top level, they can be decomposed according to the rules and the statement follows straightforwardly from the inductive hypotheses. Otherwise $\Gamma$ has only (possibly negated) atoms, $\ltri{}{}$-formulas and $\lbox{}{}$-formulas. We start a model construction with one world $w_0$ which satisfy exactly atoms in $(\Props \cap \Gamma)$ and with doxastic function $\Dox(i, w_0)(\alpha) = \maxhat \{ k' : \ltri{\agent}{k'} \alpha \in \Gamma \}$. It is easy to check that all formulas from $\LangE \cap \Gamma$ will be satisfied in this world. To satisfy also all boxed formulas in $\Gamma$ we use the inductive hypothesis for every negated box $\neg \lbox{\group}{k} \varphi$ in $\Gamma$: the application of the rule \lboxrule to that negated box should have at least one underivable denominator (with the corresponding partition $\summandf$) and by inductive hypothesis there is a pointed model $(M', w_0')$, satisfying all formulas in that denominator. Incorporating $M'$ into $M$ and defining distance from $w_0$ to $w_0'$ on $\group$ and all its subgroups as the sum of the corresponding values of $\summandf$ (and as $\plusinfty$ for the non-subgroups) we ensure that $\lbox{\group}{k} \varphi$ is falsified in $w_0$ while all the non-negated boxes in $\Gamma$ are satisfied.
\end{proof}

Thus, to check satisfiability of a formula $\varphi$ in \LogicName we can perform an exhaustive search for the closed tableaux in \CalculusName for $\{\varphi\}$, for which polynomial space w.r.t. size of $\varphi$ is sufficient.\footnote{Notice, that during the proof search we only need to store the grades appearing in the formulas and iterate through all partitions in $\partitions{J}{k}$ when applying the rule \lboxrule, which consist of at most $|\Agents|$ numbers not exceeding $k$. Therefore, the representation of the grades inside the formulas and how we measure their size does not matter, as long as it is the same for the numbers in the iterated partitions. For simplicity, we assume that any numbers occupies only constant memory and the size of the formula is measured as the number of nodes in its tree-representation, but the same reasoning can be applied to more refined measures. }

\begin{theorem}
\label{th:complexity}
Satisfiability checking in \LogicName is PSPACE-complete.
\end{theorem}
\begin{proof}{(Sketch)}
PSPACE-hardness follows from PSPACE-hardenss of logic \LDA ~\cite{InPraiseOfBeliefBases}, which is equivalent to a fragment of \LogicName. 
It can be solved in PSPACE by a proof search in \CalculusName, starting from $\{ \varphi \}$ and trying all possible rule applications. 
When trying the rule \lboxrule we iterate through all possible $Y \subseteq X$ and for each of them iterate through all partitions from $\partitions{\group}{k}$, which both can be stored using polynomial space w.r.t. $|\varphi|$, and the depth of the proof search is also bounded polynomially. 
\end{proof}



\section{Conclusion}

We have presented a proof-theoretic and complexity analysis of the notion of graded distributed belief, using a formal semantics based on graded belief bases. Following Spohn’s ranking theory \cite{Spohn1988}, we plan to study, in future work, a more general variant of the graded belief semantics based on ordinals instead of natural numbers. The notion of graded distributed belief we have defined is based on a counting view (Definition \ref{DefAlternative}). In future work, we intend to investigate a qualitative version by replacing the counting view with a qualitative perspective based on set inclusion. Last but not least, we aim to move from a static to a dynamic setting by extending our framework with a notion of graded belief base change.

\section*{Acknowledgments}

We thank the anonymous reviewers for their helpful comments and suggestions. This work is supported by the ANR project EpiRL (grant number ANR-22-CE23-0029) and European Union’s Horizon 2020 research and innovation programme under grant agreement No 101034440. 

\bibliographystyle{eptcs}
\bibliography{biblio}

\end{document}